\documentclass[numberwithinsect,a4paper,USenglish]{lipics}
\usepackage{verbatim}
\usepackage{xargs}[2008/03/08]

\makeatletter
 \theoremstyle{plain}
 \newtheorem*{lemma*}{Lemma}

\usepackage{mathpartir}
\usepackage{xypic}

\titlerunning{Conservativity of embeddings in the $\lambda\Pi$ calculus modulo rewriting} 

\author[1,2]{Ali Assaf}
\affil[1]{Inria, Paris, France}
\affil[2]{\'Ecole polytechnique, Palaiseau, France}
\authorrunning{A. Assaf} 

\Copyright{Ali Assaf}


\makeatother

\begin{document}
\global\long\def\operator#1{\mathsf{#1}}
\global\long\def\set#1{\mathcal{#1}}
\global\long\def\meta#1{\mathrm{#1}}

\global\long\def\pts#1{\lambda#1}
\global\long\def\lpi{\lambda\Pi}
\global\long\def\lpimod#1{\lpi/#1}
\global\long\def\lpimmod#1{\lpi^{-}/#1}

\global\long\def\app#1#2{#1\,#2}
\global\long\def\abs#1#2#3{\lambda#1\!:\!#2.\,#3}
\global\long\def\prod#1#2#3{\Pi#1\!:\!#2.\,#3}

\global\long\def\arr#1#2{#1\rightarrow#2}
\global\long\def\Prop{\operator{Prop}}
\global\long\def\Type{\operator{Type}}
\global\long\def\Kind{\operator{Kind}}

\global\long\def\freevars#1{\meta{FV}\left(#1\right)}
\global\long\def\subst#1#2#3{#3[#2\backslash#1]}

\global\long\def\rewrites{\leadsto}
\global\long\def\rewriterule#1#2#3{[#1]\ #2\rewrites#3}

\global\long\def\reduces{\longrightarrow}
\global\long\def\convertible{\equiv}

\global\long\def\cempty{\cdot}
\global\long\def\cdecl#1#2#3{#1,#2:#3}
\global\long\def\cconcat#1#2{#1,#2}

\newcommandx\wellformed[2][usedefault, addprefix=\global, 1=]{\meta{WF}_{#1}(#2)}
\newcommandx\typing[4][usedefault, addprefix=\global, 1=]{#2\vdash_{#1}#3:#4}

\global\long\def\ttype#1{u_{#1}}
\global\long\def\tterm#1{\varepsilon_{#1}}
\global\long\def\tsort#1{\dot{#1}}
\global\long\def\tprod#1#2#3{\dot{\pi}_{#1#2#3}}

\global\long\def\trans#1#2{\left|#1\right|_{#2}}
\global\long\def\ttrans#1#2{\left\Vert #1\right\Vert _{#2}}

\global\long\def\btrans#1{\varphi(#1)}
\global\long\def\bttrans#1{ \psi(#1)}

\global\long\def\completion#1{#1^{*}}

\global\long\def\conservative#1#2#3{\mathcal{C}_{S}(#1,#2,#3)}
\global\long\def\reducible#1#2#3{#1\models_{S}#2:#3}

\title{Conservativity of embeddings \protect\\
in the $\lambda\Pi$ calculus modulo rewriting\protect\\
(long version)}
\maketitle
\begin{abstract}
The $\lambda\Pi$ calculus can be extended with rewrite rules to embed
any functional pure type system. In this paper, we show that the embedding
is conservative by proving a relative form of normalization, thus
justifying the use of the $\lambda\Pi$ calculus modulo rewriting
as a logical framework for logics based on pure type systems. This
result was previously only proved under the condition that the target
system is normalizing. Our approach does not depend on this condition
and therefore also works when the source system is not normalizing.
\end{abstract}

\section{Introduction}

The \emph{$\lpi$ calculus modulo rewriting} is a logical framework
that extends the $\lpi$ calculus \cite{harper_framework_1993} with
rewrite rules. Through the Curry-de Bruijn-Howard correspondence,
it can express properties and proofs of various logics. Cousineau
and Dowek \cite{cousineau_embedding_2007} introduced a general embedding
of \emph{functional pure type systems} (FPTS), a large class of typed
$\lambda$-calculi, in the $\lpi$ calculus modulo rewriting: for
any FPTS $\pts S$, they constructed the system $\lpimod S$ using
appropriate rewrite rules, and defined two translation functions $\trans M{}$
and $\ttrans A{}$ that translate respectively the terms and the types
of $\pts S$ to $\lpimod S$. %
This embedding is complete, in the sense preserves typing: if $\typing[\pts S]{\Gamma}MA$
then $\typing[\lpimod S]{\ttrans{\Gamma}{}}{\trans M{}}{\ttrans A{}}$.
From the logical point of view, it preserves provability. The converse
property, called \emph{conservativity}, was only shown partially:
assuming $\lpimod S$ is strongly normalizing, if there is a term
$N$ such that $\typing[\lpimod S]{\ttrans{\Gamma}{}}N{\ttrans A{}}$
then there is a term $M$ such that $\typing[\pts S]{\Gamma}MA$.

\paragraph*{Normalization and conservativity}

Not much is known about normalization in $\lpimod S$. Cousineau and
Dowek \cite{cousineau_embedding_2007} showed that the embedding preserves
reduction: if $M\reduces M'$ then $\trans M{}\reduces^{+}\trans{M'}{}$.
As a consequence, if $\lpimod S$ is strongly normalizing (i.e.~every
well-typed term normalizes) then so is $\pts S$, but the converse
might not be true\emph{ a priori}. This was not enough to show the
conservativity of the embedding, so the proof relied on the unproven
assumption that $\lpimod S$ is normalizing. This result is insufficient
if one wants to consider the $\lpi$ calculus modulo rewriting as
a general logical framework for defining logics and expressing proofs
in those logics, as proposed in \cite{boespflug_coqine:_2012,boespflug_-calculus_2012}.
Indeed, if the embedding turns out to be inconsistent then checking
proofs in the logical framework has very little benefit.

Consider the PTS $\pts{HOL}$ that corresponds to higher order logic
\cite{barendregt_lambda_1992}:
\[
\begin{array}{ccl}
\set S & = & \Prop,\Type,\Kind\\
\set A & = & (\Prop:\Type),(\Type:\Kind)\\
\set R & = & (\Prop,\Prop,\Prop),(\Type,\Prop,\Prop),(\Type,\Type,\Type)
\end{array}
\]
This PTS is strongly normalizing, and therefore consistent. A polymorphic
variant of $\pts{HOL}$ is specified by $U^{-}=HOL+(\Kind,\Type,\Type)$.
It turns out that $\pts{U^{-}}$ is inconsistent: there is a term
$\omega$ such that $\typing[\pts{U^{-}}]{}{\omega}{\prod{\alpha}{\Prop}{\alpha}}$
and which is not normalizing \cite{barendregt_lambda_1992}. We motivate
the need for a proof of conservativity with the following example.
\begin{example}
\label{ex:polymorphic-identity} The polymorphic identity function
$I=\abs{\alpha}{\Type}{\abs x{\alpha}x}$ is \emph{not} well-typed
in $\pts{HOL}$, but it is well-typed in $\pts{U^{-}}$ and so is
its type:
\[
\typing[\pts{U^{-}}]{}I{\prod{\alpha}{\Type}{\arr{\alpha}{\alpha}}}
\]
\[
\typing[\pts{U^{-}}]{}{\prod{\alpha}{\Type}{\arr{\alpha}{\alpha}}}{\Type}
\]
However, the translation $\trans I{}=\abs{\alpha}{\ttype{\Type}}{\abs x{\app{\tterm{\Type}}{\alpha}}x}$
\emph{is} well-typed in $\lpimod{HOL}$:
\[
\typing[\lpimod{HOL}]{}{\trans I{}}{\prod{\alpha}{\ttype{\Type}}{\arr{\app{\tterm{\Type}}{\alpha}}{\app{\tterm{\Type}}{\alpha}}}}
\]
\[
\typing[\lpimod{HOL}]{}{\prod{\alpha}{\ttype{\Type}}{\arr{\app{\tterm{\Type}}{\alpha}}{\app{\tterm{\Type}}{\alpha}}}}{\Type}
\]
It seems that $\lpimod{HOL}$, just like $\pts{U^{-}}$, allows more
functions than $\pts{HOL}$, even though the type of $\trans I{}$
is not the translation of a $\pts{HOL}$ type. Is that enough to make
$\lpimod{HOL}$ inconsistent?
\end{example}

\paragraph*{Absolute normalization vs relative normalization}

One way to answer the question is to prove strong normalization of
$\lpimod S$ by constructing a model, for example in the algebra of
\emph{reducibility candidates} \cite{girard_interpretation_1972}.
Dowek \cite{dowek_models_2014} recently constructed such a model
for the embedding of higher-order logic ($\pts{HOL}$) and of the
calculus of constructions ($\pts C$). However, this technique is
still very limited. Indeed, proving such a result is, by definition,
at least as hard as proving the consistency of the original system.
It requires specific knowledge of $\pts S$ and the construction of
such a model can be very involved, such as for the calculus of constructions
with an infinite universe hierarchy ($\lambda C^{\infty}$).

In this paper, we take a different approach and show that $\lpimod S$
is conservative in all cases, even when $\pts S$ is \emph{not} normalizing.
Instead of showing that $\lpimod S$ is strongly normalizing, we show
that it is weakly normalizing \emph{relative to $\pts S$}, meaning
that proofs in the target language can be reduced to proofs in the
source language. That way we prove only what is needed to show conservativity,
without having to prove the consistency of $\pts S$ all over again.
After identifying the main difficulties, we characterize a \emph{PTS
completion} \cite{severi_pure_1994,severi_pure_1995} $\completion S$
containing $S$, and define an inverse translation from $\lpimod S$
to $\pts{\completion S}$. We then prove that $\pts{\completion S}$
is a conservative extension of $\pts S$ using the \emph{reducibility
method} \cite{tait_intensional_1967}.

\paragraph*{Outline}

The rest of the paper is organized as follows. In Section \ref{sec:pure-type-systems},
we recall the theory of pure type systems. In Section \ref{sec:lambda-Pi-calculus-modulo},
we present the framework of the $\lpi$ calculus modulo rewriting.
In Section \ref{sec:embedding}, we introduce Cousineau and Dowek's
embedding of functional pure type systems in the $\lpi$ calculus
modulo rewriting. In Section \ref{sec:conservativity}, we prove the
conservativity of the embedding using the techniques mentioned above.
In Section \ref{sec:conclusion}, we summarize the results and discuss
future work.

\section{\label{sec:pure-type-systems}Pure type systems}

Pure type systems \cite{barendregt_lambda_1992} are a general class
of typed $\lambda$-calculi parametrized by a specification.
\begin{definition}
A PTS \emph{specification} is a triple $S=(\set S,\set A,\set R)$
where
\begin{itemize}
\item $\set S$ is a set of of symbols called \emph{sorts}
\item $\set A\subseteq\set S\times S$ is a set of \emph{axioms} of the
form $(s_{1}:s_{2})$
\item $\set R\subseteq\set S\times\set S\times\set S$ is a set of \emph{rules}
of the form $(s_{1},s_{2},s_{3})$
\end{itemize}
We write $(s_{1},s_{2})$ as a short-hand for the rule $(s_{1},s_{2},s_{2})$.
The specification $S$ is \emph{functional} if the relations $\set A$
and $\set R$ are functional, that is $(s_{1},s_{2})\in\set A$ and
$(s_{1},s_{2}')\in\set A$ imply $s_{2}=s_{2}'$, and $(s_{1},s_{2},s_{3})\in\set R$
and $(s_{1},s_{2},s_{3}')\in\set R$ imply $s_{3}=s_{3}'$. The specification
is \emph{full} if for all $s_{1},s_{2}\in\set S$, there is a sort
$s_{3}$ such that $(s_{1},s_{2},s_{3})\in\set R$.
\end{definition}

\begin{definition}
Given a PTS specification $S=(\set S,\set A,\set R)$ and a countably
infinite set of variables $\set V$, the abstract syntax of $\pts S$
is defined by the following grammar: 
\[
\begin{array}{lrcl}
\mbox{(terms)} & \set T & ::= & \set S\mid\set V\mid\app{\set T}{\set T}\mid\abs{\set V}{\set T}{\set T}\mid\prod{\set V}{\set T}{\set T}\\
\mbox{(contexts)} & \set C & ::= & \cempty\mid\cdecl{\set C}{\set V}{\set T}
\end{array}
\]
We use lower case letters $x,y,\alpha,\beta,\ldots$ to denote variables,
uppercase letters such as $M,N$, $A,B,\ldots$ to denote terms, and
uppercase Greek letters such as $\Gamma,\Delta,\Sigma,\ldots$ to
denote contexts. The set of free variables of a term $M$ is denoted
by $\meta{FV}\left(M\right)$. We write $\arr AB$ for $\prod xAB$
when $x\not\in\freevars B$.

The typing rules of $\pts S$ are presented in Figure \ref{fig:typing-rules-pts}.
We write $\typing{\Gamma}MA$ instead of $\typing[\pts S]{\Gamma}MA$
when the context is unambiguous. We say that $M$ is a \emph{$\Gamma$-term}
when $\wellformed{\Gamma}$ and $\typing{\Gamma}MA$ for some $A$.
We say that $A$ is a \emph{$\Gamma$-type} when $\wellformed{\Gamma}$
and either $\typing{\Gamma}As$ or $A=s$ for some $s\in\set S$.
We write $\typing{\Gamma}M{A:s}$ as a shorthand for $\Gamma\vdash M:A\wedge\Gamma\vdash A:s$.

\begin{figure}
\begin{mathpar}
\infer[Empty]{ }{
	\wellformed[\pts{S}]{\cempty}}

\infer[Declaration]{
	\typing[\pts{S}]{\Gamma}{A}{s} \\
	x\not\in\Gamma}{
	\wellformed[\pts{S}]{\cdecl{\Gamma}{x}{A}}}

\infer[Variable]{
	\wellformed[\pts{S}]{\Gamma} \\
	(x:A)\in\Gamma}{
	\typing[\pts{S}]{\Gamma}{x}{A}}

\infer[Sort]{
	\wellformed[\pts{S}]{\Gamma} \\
	(s_1:s_2)\in{\set{A}}}{
	\typing[\pts{S}]{\Gamma}{s_1}{s_2}}

\infer[Product]{
	\typing[\pts{S}]{\Gamma}{A}{s_1} \\
	\typing[\pts{S}]{\cdecl{\Gamma}{x}{A}}{B}{s_2} \\
	(s_1,s_2,s_3)\in{\set{R}}}{
	\typing[\pts{S}]{\Gamma}{\prod{x}{A}{B}}{s_3}}

\infer[Abstraction]{
	\typing[\pts{S}]{\cdecl{\Gamma}{x}{A}}{M}{B} \\
	\typing[\pts{S}]{\Gamma}{\prod{x}{A}{B}}{s}}{
	\typing[\pts{S}]{\Gamma}{\abs{x}{A}{M}}{\prod{x}{A}{B}}}

\infer[Application]{
	\typing[\pts{S}]{\Gamma}{M}{\prod{x}{A}{B}} \\
	\typing[\pts{S}]{\Gamma}{N}{A}}{
	\typing[\pts{S}]{\Gamma}{\app{M}{N}}{\subst{N}{x}{B}}}

\infer[Conversion]{
	\typing[\pts{S}]{\Gamma}{M}{A} \\
	\typing[\pts{S}]{\Gamma}{B}{s} \\
	A\equiv_\beta B}{
	\typing[\pts{S}]{\Gamma}{M}{B}}
\end{mathpar}

\protect\caption{\label{fig:typing-rules-pts}Typing rules of $\protect\pts S$}
\end{figure}
\end{definition}
\begin{example}
\label{ex:pure-type-systems}The following well-known systems can
all be expressed as functional pure type systems using the same set
of sorts $\set S=\Type,\Kind$ and the same set of axioms $\set A=(\Type:\Kind)$:
\begin{itemize}
\item Simply-typed $\lambda$ calculus ($\pts{\!\rightarrow}$): \\
$\set R=(\Type,\Type)$
\item System F ($\pts 2$): \\
$\set R=(\Type,\Type),(\Kind,\Type)$
\item $\lambda\Pi$ calculus ($\pts P$): \\
$\set R=(\Type,\Type),(\Type,\Kind)$
\item Calculus of constructions ($\pts C$): \\
$\set R=(\Type,\Type),(\Kind,\Type),(\Type,\Kind),(\Kind,\Kind)$
\end{itemize}
\end{example}

\begin{example}
Let $I=\abs{\alpha}{\Type}{\abs x{\alpha}x}$ be the polymorphic identity
function. The term $I$ is not well-typed in the simply typed $\lambda$
calculus but it is well-typed in the calculus of constructions $\pts C$:
\[
\typing[\pts C]{}I{\prod{\alpha}{\Type}{\arr{\alpha}{\alpha}}}
\]

\end{example}
The following properties hold for all pure type systems \cite{barendregt_lambda_1992}.
\begin{theorem}[Correctness of types]
 If $\typing[\pts S]{\Gamma}MA$ then $\wellformed[\pts S]{\Gamma}$
and either $\typing[\pts S]{\Gamma}As$ or $A=s$ for some $s\in\set S$,
i.e. $A$ is a $\Gamma$-type.
\end{theorem}
The reason why we don't always have $\typing[\pts S]{\Gamma}As$ is
that some sorts do not have an associated axiom, such as $\Kind$
in Example \ref{ex:pure-type-systems}, which leads to the following
definition.
\begin{definition}[Top-sorts]
 A sort $s\in\set S$ is called a \emph{top-sort} when there is no
sort $s'\in\set S$ such that $(s:s')\in\set A$.
\end{definition}
The following property is useful for proving properties about systems
with top-sorts.
\begin{theorem}[Top-sort types]
\label{thm:pts-top-sort-types} If $\typing[\pts S]{\Gamma}As$ and
$s$ is a top-sort then either $A=s'$ for some sort $s'\in\set S$
or $A=\prod xBC$ for some terms $B,C$.
\end{theorem}

\begin{theorem}[Confluence]
 If $M_{1}\reduces_{\beta}^{*}M_{2}$ and $M_{1}\reduces_{\beta}^{*}M_{3}$
then there is a term $M_{4}$ such that $M_{2}\reduces_{\beta}^{*}M_{4}$
and $M_{3}\reduces_{\beta}^{*}M_{4}$.
\end{theorem}

\begin{theorem}[Product compatibility]
 If $\prod xA{B\equiv_{\beta}\prod x{A'}{B'}}$ then $A\equiv_{\beta}A'$
and $B\equiv_{\beta}B'$.
\end{theorem}

\begin{theorem}[Subject reduction]
 If $\typing[\pts S]{\Gamma}MA$ and $M\reduces_{\beta}^{*}M'$ then
$\typing[\pts S]{\Gamma}{M'}A$.
\end{theorem}
Finally, we state the following property for functional pure type
systems.
\begin{theorem}[Uniqueness of types]
 Let $\set S$ be a functional specification. If $\typing[\pts S]{\Gamma}MA$
and $\typing[\pts S]{\Gamma}MB$ then $A\equiv_{\beta}B$.
\end{theorem}
In the rest of the paper, all the pure type systems we will consider
will be functional.

\section{\label{sec:lambda-Pi-calculus-modulo}The $\protect\lpi$ calculus
modulo rewriting}

The\emph{ $\lpi$ calculus}, also known as $LF$ and as $\pts P$,
is one of the simplest forms of $\lambda$ calculus with dependent
types, and corresponds through the Curry-de Bruijn-Howard correspondence
to a minimal first-order logic of higher-order terms. As mentioned
in Example \ref{ex:pure-type-systems}, it can be defined as the functional
pure type system $\pts P$ with the following specification:
\[
\begin{array}{rcl}
\set S & = & \Type,\Kind\\
\set A & = & \Type:\Kind\\
\set R & = & (\Type,\Type),(\Type,\Kind)
\end{array}
\]

The\emph{ $\lpi$ calculus modulo rewriting} extends the $\lpi$ calculus
with rewrite rules. By equating terms modulo a set of rewrite rules
$R$ in addition to $\alpha$ and $\beta$ equivalence, it can type
more terms using the conversion rule, and therefore express theories
that are more complex. The calculus can be seen as a variant of Martin-L\"{o}f's
logical framework \cite{nordstrom_programming_1990,luo_computation_1994}
where equalities are expressed as rewrite rules.

We recall that a rewrite rule is a triple $\rewriterule{\Delta}MN$
where $\Delta$ is a context and $M,N$ are terms such that $\freevars N\subseteq\freevars M$.
A set of rewrite rules $R$ induces a reduction relation on terms,
written $\reduces_{R}$, defined as the smallest contextual closure
such that if $\rewriterule{\Delta}MN\in R$ then $\sigma(M)\reduces_{R}\sigma(N)$
for any substitution $\sigma$ of the variables in $\Delta$. We define
the relation $\reduces_{\beta R}$ as $\reduces_{\beta}\cup\reduces_{R}$,
the relation $\equiv_{R}$ as the smallest congruence containing $\reduces_{R}$,
and the relation $\equiv_{\beta R}$ as the smallest congruence containing
$\reduces_{\beta R}$.
\begin{definition}
A rewrite rule $\rewriterule{\Delta}MN$ is \emph{well-typed in a
context $\Sigma$} when there is a term $A$ such that $\typing[\lpi]{\Sigma,\Delta}MA$
and $\typing[\lpi]{\cconcat{\Sigma}{\Delta}}NA$.
\end{definition}

\begin{definition}
Let $\Sigma$ be a well-formed $\lpi$ context and $R$ a set of rewrite
rules that are well-typed in $\Sigma$. The \emph{$\lambda\Pi$ calculus
modulo $(\Sigma,R)$}, written $\lpimod{(\Sigma,R)}$, is defined
with the same syntax as the $\lambda\Pi$ calculus, but with the typing
rules of Figure \ref{fig:typing-rules-lpimod}. We write $\lpi/$
instead of $\lpimod{(\Sigma,R)}$ when the context is unambiguous.
\end{definition}
\begin{figure}
\begin{mathpar}
\infer[Empty]{ }{
	\wellformed[\lpimod{}]{\cempty}}

\infer[Declaration]{
	\typing[\lpimod{}]{\Gamma}{A}{s} \\
	x\not\in\cconcat{\Sigma}{\Gamma}}{
	\wellformed[\lpimod{}]{\cdecl{\Gamma}{x}{A}}}

\infer[Variable]{
	\wellformed[\lpimod{}]{\Gamma} \\
	(x:A)\in\cconcat{\Sigma}{\Gamma}}{
	\typing[\lpimod{}]{\Gamma}{x}{A}}

\infer[Sort]{
	\wellformed[\lpimod{}]{\Gamma} \\
	(s_1:s_2)\in{\set{A}}}{
	\typing[\lpimod{}]{\Gamma}{s_1}{s_2}}

\infer[Product]{
	\typing[\lpimod{}]{\Gamma}{A}{s_1} \\
	\typing[\lpimod{}]{\cdecl{\Gamma}{x}{A}}{B}{s_2} \\
	(s_1,s_2,s_3)\in{\set{R}}}{
	\typing[\lpimod{}]{\Gamma}{\prod{x}{A}{B}}{s_3}}

\infer[Abstraction]{
	\typing[\lpimod{}]{\cdecl{\Gamma}{x}{A}}{M}{B} \\
	\typing[\lpimod{}]{\Gamma}{\prod{x}{A}{B}}{s}}{
	\typing[\lpimod{}]{\Gamma}{\abs{x}{A}{M}}{\prod{x}{A}{B}}}

\infer[Application]{
	\typing[\lpimod{}]{\Gamma}{M}{\prod{x}{A}{B}} \\
	\typing[\lpimod{}]{\Gamma}{N}{A}}{
	\typing[\lpimod{}]{\Gamma}{\app{M}{N}}{\subst{N}{x}{B}}}

\infer[Conversion]{
	\typing[\lpimod{}]{\Gamma}{M}{A} \\
	\typing[\lpimod{}]{\Gamma}{B}{s} \\
	A\equiv_{\beta R}B}{
	\typing[\lpimod{}]{\Gamma}{M}{B}}
\end{mathpar}

\protect\caption{\label{fig:typing-rules-lpimod}Typing rules of $\protect\lpimod{(\Sigma,R)}$}
\end{figure}

\begin{example}
Let $\Sigma$ be the context 
\[
\alpha:\Type,c:\alpha,f:\arr{\alpha}{\Type}
\]
 and $R$ be the following rewrite rule 
\[
\rewriterule{\cempty}{\app fc}{\arr{\prod y{\alpha}{\app fy}}{\app fy}}
\]
 Then the term 
\[
\delta=\abs x{\app fc}{\app{\app xc}x}
\]
 is well-typed in $\lpimod{(\Sigma,R)}$:
\[
\typing[\lpimod{(\Sigma,R)}]{}{\delta}{\arr{\app fc}{\app fc}}
\]
Note that the term $\delta$ would not be well-typed without the rewrite
rule, even if we replace all the occurrences of $\app fc$ in $\delta$
by $\arr{\prod y{\alpha}{\app fy}}{\app fy}$.
\end{example}
The system $\lpi$ is a pure type system and therefore enjoys all
the properties mentioned in Section \ref{sec:pure-type-systems}.
The behavior of $\lpimod{\left(\Sigma,R\right)}$ however depends
on the choice of $\left(\Sigma,R\right)$. In particular, some properties
analogous to those of pure type systems depend on the confluence of
the relation $\reduces_{\beta R}$.
\begin{theorem}[Correctness of types]
 If $\typing[\lpimod{}]{\Gamma}MA$ then $\wellformed[\lpimod{}]{\Gamma}$
and either $\typing[\lpimod{}]{\Gamma}As$ for some $s\in\left\{ \Type,\Kind\right\} $
or $A=\Kind$.
\end{theorem}

\begin{theorem}[Top-sort types]
 If $\typing[\lpimod{}]{\Gamma}A{\Kind}$ then either $A=\Type$
or $A=\prod xBC$ for some terms $B,C$ such that $\typing[\lpimod{}]{\Gamma,x:B}C{\Kind}$.
\end{theorem}
Assuming $\reduces_{\beta R}$ is confluent, the following properties
hold \cite{blanqui_definitions_2005}.
\begin{theorem}[Product compatibility]
 If $\prod xA{B\equiv_{\beta R}\prod x{A'}{B'}}$ then $A\equiv_{\beta R}A'$
and $B\equiv_{\beta R}B'$.
\end{theorem}

\begin{theorem}[Subject reduction]
 If $\typing[\lpimod{}]{\Gamma}MA$ and $M\reduces_{\beta R}^{*}M'$
then $\typing[\lpimod{}]{\Gamma}{M'}A$.
\end{theorem}

\begin{theorem}[Uniqueness of types]
 If $\typing[\lpimod{}]{\Gamma}MA$ and $\typing[\lpimod{}]{\Gamma}MB$
then $A\equiv_{\beta R}B$.
\end{theorem}

\section{\label{sec:embedding}Embedding FPTS's in the $\lambda\Pi$ calculus
modulo}

In this section, we present the embedding of functional pure type
systems in the $\lambda\Pi$ calculus modulo rewriting as introduced
by Cousineau and Dowek \cite{cousineau_embedding_2007}. In this embedding,
sorts are represented as \emph{universes \`{a} la Tarski}, as introduced
by Martin-L\"{o}f \cite{martin-lof_intuitionistic_1984} and later
developed by Luo \cite{luo_computation_1994} and Palmgren \cite{palmgren_universes_1998}.
The embedding is done in two steps. First, given a pure type system
$\pts S$, we construct $\lpimod S$ by giving an appropriate signature
and rewrite system. Second, we define a translation from the terms
and types of $\pts S$ to the terms and types of $\lpimod S$. The
proofs of the theorems in this section can be found in the original
paper \cite{cousineau_embedding_2007}.
\begin{definition}[The system $\lpimod S$]
 Consider a functional pure type system specified by $S=(\set S,\set A,\set R)$.
Define $\Sigma_{S}$ to be the well-formed context containing the
declarations:
\[
\begin{array}{ll}
\ttype s:\Type & \forall s\in\set S\\
\tterm s:\arr{\ttype s}{\Type} & \forall s\in\set S\\
\tsort{s_{1}}:\ttype{s_{2}} & \forall s_{1}:s_{2}\in\set A\\
\tprod{s_{1}}{s_{2}}{s_{3}}:\prod{\alpha}{\ttype{s_{1}}}{\arr{(\arr{\app{\tterm{s_{1}}}{\alpha}}{\ttype{s_{2}}})}{\ttype{s_{3}}}} & \forall(s_{1},s_{2},s_{3})\in\set R
\end{array}
\]
Let $R_{S}$ be the well-typed rewrite system containing the rules
\[
\rewriterule{\cempty}{\app{\tterm{s_{2}}}{\tsort{s_{1}}}}{\ttype{s_{1}}}
\]
 for all $s_{1}:s_{2}\in\set A$, and 
\[
\rewriterule{\Delta_{s_{1}s_{2}s_{3}}}{\app{\tterm{s_{3}}}{(\app{\app{\tprod{s_{1}}{s_{2}}{s_{3}}}A}B)}}{\prod x{(\app{\tterm{s_{1}}}A)}{\app{\tterm{s_{2}}}{(\app Bx)}}}
\]
 for all $(s_{1},s_{2},s_{3})\in\set R$, where $\Delta_{s_{1}s_{2}s_{3}}=(A:\ttype{s_{1}},B:(\arr{\app{\tterm{s_{1}}}{\alpha}}{\ttype{s_{2}}}))$.
The system $\lpimod S$ is defined as the $\lambda\Pi$ calculus modulo
$(\Sigma_{S},R_{S})$, that is, $\lpimod{(\Sigma_{S},R_{S}})$.\end{definition}
\begin{theorem}[Confluence]
 The relation $\reduces_{\beta R}$ is confluent.
\end{theorem}
The translation is composed of two functions, one from the terms of
$\pts S$ to the terms of $\lpimod S$, the other from the types of
$\pts S$ to the types of $\lpimod S$.
\begin{definition}
The translation $\trans M{\Gamma}$ of $\Gamma$-terms and the translation
$\ttrans A{\Gamma}$ of $\Gamma$-types are mutually defined as follows.
\[
\begin{array}{rcl}
\trans s{\Gamma} & = & \tsort s\\
\trans x{\Gamma} & = & x\\
\trans{\app MN}{\Gamma} & = & \app{\trans M{\Gamma}}{\trans N{\Gamma}}\\
\trans{\abs xAM}{\Gamma} & = & \abs x{\ttrans A{\Gamma}}{\trans M{\cdecl{\Gamma}xA}}\\
\trans{\prod xAB}{\Gamma} & = & \app{\app{\tprod{s_{1}}{s_{2}}{s_{3}}}{\trans A{\Gamma}}}{(\abs x{\ttrans A{\Gamma}}{\trans B{\cdecl{\Gamma}xA}})}\\
 &  & \mbox{where \ensuremath{\typing{\Gamma}A{s_{1}}}}\\
 &  & \mbox{and \ensuremath{\typing{\cdecl{\Gamma}xA}B{s_{2}}}}\\
 &  & \mbox{and \ensuremath{(s_{1},s_{2},s_{3})\in\set R}}
\end{array}
\]
\[
\begin{array}{rcl}
\ttrans s{\Gamma} & = & \ttype s\\
\ttrans{\prod xAB}{\Gamma} & = & \prod x{\ttrans A{\Gamma}}{\ttrans B{\cdecl{\Gamma}xA}}\\
\ttrans A{\Gamma} & = & \app{\tterm s}{\trans A{\Gamma}}\mbox{ where \ensuremath{\typing{\Gamma}As}}
\end{array}
\]
Note that this definition is redundant but it is well-defined up to
$\convertible_{\beta R}$. In particular, because some $\Gamma$-types
are also $\Gamma$-terms, there are two ways to translate them, but
they are equivalent:
\[
\begin{array}{rcl}
\app{\tterm{s_{2}}}{\tsort{s_{1}}} & \convertible_{\beta R} & \ttype{s_{1}}\\
\app{\tterm{s_{3}}}{\trans{\prod xAB}{\Gamma}} & \convertible_{\beta R} & \prod x{\ttrans A{\Gamma}}{\ttrans B{\cdecl{\Gamma}xA}}
\end{array}
\]
This definition is naturally extended to well-formed contexts as follows.
\[
\begin{array}{rcl}
\ttrans{\cempty}{} & = & \cempty\\
\ttrans{\cdecl{\Gamma}xA}{} & = & \cdecl{\ttrans{\Gamma}{}}x{\ttrans A{\Gamma}}
\end{array}
\]
\end{definition}
\begin{example}
The polymorphic identity function of the Calculus of constructions
$\pts C$ is translated as
\[
\trans I{}=\abs{\alpha}{\ttype{\Type}}{\abs x{\app{\tterm{\Type}}{\alpha}}x}
\]
and its type $A=\prod{\alpha}{\Type}{\arr{\alpha}{\alpha}}$ is translated
as: 
\[
\trans A{}=\app{\app{\tprod{\Kind}{,\Type,}{\Type}}{\tsort{\Type}}}{(\abs{\alpha}{\ttype{\Type}}{\trans{A_{\alpha}}{}})}
\]
 where $A_{\alpha}=\arr{\alpha}{\alpha}$ and
\[
\trans{A_{\alpha}}{}=\app{\app{\tprod{\Type}{,\Type,}{\Type}}{\alpha}}{(\abs x{\app{\tterm{Type}}{\alpha}}{\app{\tterm{Type}}{\alpha}})}
\]
The identity function applied to itself is translated as:
\[
\trans{\app{\app IA}I}{}=\app{\app{\trans I{}}{\trans A{}}}{\trans I{}}
\]

\end{example}
The embedding is complete, in the sense that all the typing relations
of $\pts S$ are preserved by the translation.
\begin{theorem}[Completeness]
\label{thm:completeness} For any context $\Gamma$ and terms $M$
and $A$, if $\typing[\pts S]{\Gamma}MA$ then $\typing[\lpimod S]{\ttrans{\Gamma}{}}{\trans M{\Gamma}}{\ttrans A{\Gamma}}$.
\end{theorem}

\section{\label{sec:conservativity}Conservativity}

In this section, we prove the converse of the completeness property.
One could attempt to prove that if $\typing[\lpimod S]{\ttrans{\Gamma}{}}{\trans M{\Gamma}}{\ttrans A{\Gamma}}$
then $\typing[\pts S]{\Gamma}MA$. However, that would be too weak
because the translation $\trans M{\Gamma}$ is only defined for well-typed
terms. A second attempt would be to define inverse translations $\btrans M$
and $\bttrans A$ and prove that if $\typing[\lpimod S]{\Gamma}MA$
then $\typing[\pts S]{\bttrans{\Gamma}}{\btrans M}{\bttrans A}$,
but that would not work either because not all terms and types of
$\lpimod S$ correspond to valid terms and types of $\pts S$, as
was shown in Example \ref{ex:polymorphic-identity}. Therefore the
property that we want to prove is: if there is a term $N$ such that
$\typing[\lpimod S]{\ttrans{\Gamma}{}}N{\ttrans A{\Gamma}}$ then
there is a term $M$ such that $\typing[\pts S]{\Gamma}MA$.

The main difficulty is that some of these \emph{external} terms can
be involved in witnessing valid $\pts S$ types, as illustrated by
the following example.
\begin{example}
Consider the context $nat:\Type$. Even though the polymorphic identity
function $I$ and its type are not well-typed in $\pts{HOL}$, they
can be used in $\lpimod{HOL}$ to construct a witness for $\arr{nat}{nat}$.
\[
\typing[\lpimod{HOL}]{nat:\ttype{\Type}}{(\app{\trans I{}}{nat})}{(\arr{\app{\tterm{\Type}}{nat}}{\app{\tterm{\Type}}{nat}})}
\]
We can normalize the term $\app{\trans I{}}{nat}$ to $\abs x{\app{\tterm{\Type}}{nat}}x$
which is a term that corresponds to a valid $\pts{HOL}$ term: it
is the translation of the term $\abs x{nat}x$. However, as discussed
previously, we cannot restrict ourselves to normal terms because we
do not know if $\lpimod S$ is normalizing.
\end{example}
To prove conservativity, we will therefore need to address the following
issues:
\begin{enumerate}
\item The system $\lpimod S$ can type more terms than $\pts S$.
\item These terms can be used to construct proofs for the translation of
$\pts S$ types.
\item The $\lpimod S$ terms that inhabit the translation of $\pts S$ types
can be reduced to the translation of $\pts S$ terms.
\end{enumerate}
We will proceed as follows. First, we will eliminate $\beta$-redexes
at the level of $\Kind$ by reducing $\lpimod S$ to a subset $\lpimmod S$.
Then, we will extend $\pts S$ to a \emph{minimal completion} $\pts{S^{*}}$
that can type more terms than $\pts S$, and show that $\lpimmod S$
corresponds to $\pts{S^{*}}$ using inverse translations $\btrans M$
and $\bttrans A$. Finally, we will show that $\pts{S^{*}}$ terms
inhabiting $\pts S$ types can be reduced to $\pts S$ terms. The
procedure is summarized in the following diagram.

\[\xymatrix@R=4pc@C=4pc{
 \lpimod{S}\ar[r]^{\text{(Lemma \ref{lem:kind-redex-elimination})}}_{\beta^*} & \lpimmod{S}\ar[d]_{\btrans{M}}^{\bttrans{A}\ \text{(Lemma \ref{lem:inverse-completion})}} \\
 \pts{S} \ar@{..>}[u]_{\ttrans{A}{}}^{\text{(Theorem \ref{thm:completeness})}\ \trans{M}{}} & \pts{S^*} \ar[l]^{\beta^*}_{\text{(Lemma \ref{lem:completion-reducible})}} \\
}\]

\subsection{Eliminating $\beta$-redexes at the level of $\protect\Kind$}

In $\lpimod S$, we can have $\beta$-redexes at the level of $\Kind$
such as $\app{(\abs xA{\ttype s})}M$. These redexes are artificial
and are never generated by the forward translation of any PTS. We
show here that they can always be safely eliminated.
\begin{definition}
A $\Gamma$-term $M$ of type $C$ is at the level of $\Kind$ (resp.
$\Type$) if $\typing{\Gamma}C{\Kind}$ (resp. $\typing{\Gamma}C{\Type}$).
We define $\lpimmod S$ terms as the subset of well-typed $\lpimod S$
terms that do not contain any $\Kind$-level $\beta$-redexes.\end{definition}
\begin{lemma}
\label{lem:kind-redex-elimination} For any $\lpimod S$ context $\Gamma$
and $\Gamma$-term $M$, there is a $\lpimmod S$ term $M^{-}$ such
that $M\reduces_{\beta}^{*}M^{-}$.\end{lemma}
\begin{proof}
Reducing a $\Kind$-level $\beta$-redex $\app{(\abs xAB)}N$ does
not create other $\Kind$-level $\beta$-redexes because $N$ is at
the level of $\Type$. Indeed, in the $\lambda\Pi$ calculus modulo
rewriting the only $\Kind$ rule is $\left(\Type,\Kind,\Kind\right)$.
Therefore $N:A:\Type$. If $N$ reduces to a $\lambda$-abstraction
then the only redexes it can create are at the level of $\Type$.
Therefore, the number of $\Kind$-level $\beta$-redexes strictly
decreases, so any $\Kind$-level $\beta$-reduction strategy will
terminate.\end{proof}
\begin{example}
The term
\[
I_{1}=\abs{\alpha}{\ttype{\Type}}{\abs x{\app{\tterm{\Type}}{(\app{(\abs{\beta}{\ttype{\Type}}{\beta})}{\alpha})}}x}
\]
is in $\lpimmod{HOL}$. The term 
\[
I_{2}=\abs{\alpha}{\ttype{\Type}}{\abs x{(\app{(\abs{\beta}{\ttype{\Type}}{\app{\tterm{\Type}}{\beta}})}{\alpha})}x}
\]
 is not in $\lpimmod{HOL}$ but 
\[
I_{2}\reduces_{\beta}\abs{\alpha}{\ttype{\Type}}{\abs x{\tterm{Type}\alpha}x}
\]
 which is in $\lpimmod{HOL}$.
\end{example}

\subsection{Minimal completion}

To simplify our reducibility proof in the next section, we will translate
$\lpimod S$ back to a pure type system, but since it cannot be $\pts S$
we will define a slightly larger PTS called $\pts{S^{*}}$ that contains
$\pts S$ and that will be easier to manipulate than $\lpimod S$.

The reason we need a larger PTS is that we have types that do not
have a type, such as top-sorts because there is no associated axiom.
Similarly, we can sometimes prove $\typing[\pts S]{\cdecl{\Gamma}xA}MB$
but cannot abstract over $x$ because there is no associated product
rule. Completions of pure type systems were originally introduced
by Severi \cite{severi_pure_1994,severi_pure_1995} to address these
issues by injecting $\pts S$ into a larger pure type system.
\begin{definition}[Completion \cite{severi_pure_1995}]
\label{def:completion}  A specification $S'=(\set S',\set A',\mbox{\ensuremath{\set R}}')$
is a \emph{completion of} $S$ if
\begin{enumerate}
\item $\set S\subseteq\set S'$,$\set A\subseteq\set A'$, $\set R\subseteq\set R'$,
and
\item for all sorts $s_{1}\in\set S$, there is a sort $s_{2}\in\set S'$
such that $(s_{1}:s_{2})\in\set A'$, and
\item for all sorts $s_{1},s_{2}\in\set S'$, there is a sort $s_{3}\in\set S'$
such that $(s_{1},s_{2},s_{3})\in\set R'$.
\end{enumerate}
\end{definition}
Notice that all the top-sorts of $\pts S$ are typable in $\pts{S'}$
and that $\pts{S'}$ is full, meaning that all products are typable.
These two properties reflect exactly the discrepancy between $\pts S$
and $\lpimmod S$. Not all completions are conservative though, so
we define the following completion.
\begin{definition}[Minimal completion]
\label{def:minimal-completion}  We define the \emph{minimal completion
of} $S$, written $\completion S$, to be the following specification:
\begin{eqnarray*}
\completion{\set S} & = & \set S\cup\{\tau\}\\
\completion{\set A} & = & \set A\cup\{(s_{1}:\tau)\mid s_{1}\in\set S,\nexists s2,(s_{1}:s_{2})\in\set A\}\\
\completion{\set R} & = & \set R\cup\{(s_{1},s_{2},\tau)\mid s_{1},s_{2}\in\completion{\set S},\nexists s_{3},(s_{1},s_{2},s_{3})\in\set R\}
\end{eqnarray*}
 where $\tau\not\in\set S$.
\end{definition}
We add a new top-sort $\tau$ and axioms $s:\tau$ for all previous
top-sorts $s$, and complete the rules to obtain a PTS full. The new
system is a completion by Definition \ref{def:completion} and it
is minimal in the sense that we generically added the smallest number
of sorts, axioms, and rules so that the result is guaranteed to be
conservative. Any well-typed term of $\pts S$ is also well-typed
in $\pts{\completion S}$, but just like $\lpimmod S$, this system
allows more functions than $\pts S$.
\begin{example}
The polymorphic identity function is well-typed in $\pts{\completion{HOL}}$.
\[
\vdash_{\pts{\completion{HOL}}}I:\prod{\alpha}{\Type}{\arr{\alpha}{\alpha}}
\]
\[
\vdash_{\pts{\completion{HOL}}}\prod{\alpha}{\Type}{\arr{\alpha}{\alpha}}:\tau
\]

\end{example}
Next, we define inverse translations that translate the terms and
types of $\lpimmod S$ to the terms and types of $\pts{\completion S}$.
\begin{definition}[Inverse translations]
\label{def:inverse-translation}  The inverse translation of terms
$\btrans M$ and the inverse translation of types $\bttrans A$ are
mutually defined as follows.
\[
\begin{array}{rcl}
\btrans{\tsort s} & = & s\\
\btrans{\tprod{s_{1}}{s_{2}}{s_{3}}} & = & \abs{\alpha}{s_{1}}{\abs{\beta}{(\arr{\alpha}{s_{2}})}{\prod x{\alpha}{\app{\beta}x}}}\\
\btrans x & = & x\\
\btrans{\app MN} & = & \app{\btrans M}{\btrans N}\\
\btrans{\abs xAM} & = & \abs x{\bttrans A}{\btrans M}
\end{array}
\]
\[
\begin{array}{rcl}
\bttrans{\ttype s} & = & s\\
\bttrans{\app{\tterm s}M} & = & \btrans M\\
\bttrans{\prod xAB} & = & \prod x{\bttrans A}{\bttrans B}
\end{array}
\]

\end{definition}
Note that this is only a partial definition, but it is total for $\lpimmod S$
terms. In particular, it is an inverse of the forward translation
in the following sense.
\begin{lemma}
\label{lem:inverse-translation} For any $\Gamma$-term $M$ and $\Gamma$-type
$A$,
\begin{enumerate}
\item $\btrans{\trans M{\Gamma}}\equiv_{\beta}M$,
\item $\bttrans{\ttrans A{\Gamma}}\equiv_{\beta}A$.
\end{enumerate}
\end{lemma}
\begin{proof}
By induction on $M$ or $A$. We show the product case where $M=\prod xAB$.
By induction hypothesis, $\btrans{\trans A{}}\equiv_{\beta}A$ and
$\btrans{\trans B{}}\equiv_{\beta}B$. Therefore 
\[
\begin{array}{cll}
\btrans{\trans M{}} & = & \app{\app{\left(\lambda\alpha.\,\lambda\beta.\,\prod x{\alpha}{\app{\beta}x}\right)}{\btrans{\trans A{}}}}{(\lambda x.\,\btrans{\trans B{}})}\\
 & \reduces_{\beta}^{*} & \prod x{\btrans{\trans A{}}}{\btrans{\trans B{}}}\\
 & \equiv_{\beta} & \prod xAB
\end{array}
\]

\end{proof}
Next we show that the inverse translations preserve typing.
\begin{lemma}
\label{lem:inverse-substitution} $ $
\begin{enumerate}
\item $\btrans{\subst NxM}=\subst{\btrans N}x{\btrans M}$
\item $\bttrans{\subst NxA}=\subst{\btrans N}x{\bttrans A}$
\end{enumerate}
\end{lemma}
\begin{proof}
By induction on $M$ or $A$. We show the product case $A=\prod yBC$.
Without loss of generality, $y\not=x$ and $y\not\in N$ and $y\not\in\btrans N$.
Then $\subst Nx{\prod yBC}=\prod y{\subst NxB}{\subst NxC}$. By induction
hypothesis, $\bttrans{\subst NxB}=\subst{\btrans N}x{\bttrans B}$
and $\bttrans{\subst NxC}=\subst{\btrans N}x{\bttrans C}$. Therefore
\[
\begin{array}{rcl}
\bttrans{\subst NxA} & = & \prod y{\subst{\btrans N}x{\bttrans B}}{\subst{\btrans N}x{\bttrans C}}\\
 & = & \subst{\btrans N}x{\prod x{\bttrans B}{\bttrans C}}\\
 & = & \subst{\btrans N}x{\bttrans{\prod xBC}}
\end{array}
\]
\end{proof}
\begin{lemma}
\label{lem:inverse-reduction} $ $
\begin{enumerate}
\item If $M\reduces_{\beta R}N$ then $\btrans M\reduces_{\beta}^{*}\btrans N$
\item If $A\reduces_{\beta R}B$ then $\bttrans A\reduces_{\beta}^{*}\bttrans B$
\end{enumerate}
\end{lemma}
\begin{proof}
By induction on $M$ or $A$. We show the base cases.
\begin{itemize}
\item Case $M=\app{(\abs x{A_{1}}{M_{1}})}{N_{1}}$, $N=\subst{N_{1}}x{M_{1}}$.
Then $\btrans M=\app{(\abs x{\bttrans{A_{1}}}{\btrans{M_{1}}})}{\btrans{N_{1}}}$.
Therefore $\btrans M\reduces_{\beta}\subst{\btrans{N_{1}}}x{\btrans{M_{1}}}$
which is equal to $\btrans{\subst{N_{1}}x{M_{1}}}$ by Lemma \ref{lem:inverse-substitution}.
\item Case $A=\app{\tterm s}{\tsort s}$, $B=\ttype s$. Then $\bttrans A=s=\bttrans B$.
\item Case $A=\app{\tterm{s_{1}}}{\left(\app{\app{\tprod{s_{1}}{s_{2}}{s_{3}}}{A_{1}}}{B_{1}}\right)}$,
$B=\prod x{\app{\tterm{s_{1}}}{A_{1}}}{\app{\tterm{s_{2}}}{\left(\app{B_{1}}x\right)}}$.
Then
\[
\begin{array}{lll}
\bttrans A & = & \app{\app{\left(\lambda\alpha.\,\lambda\beta.\,\prod x{\alpha}{\app{\beta}x}\right)}{\btrans{A_{1}}}}{\btrans{B_{1}}}\\
 & \reduces_{\beta}^{*} & \prod x{\btrans{A_{1}}}{\app{\btrans{B_{1}}}x}\\
 & = & \bttrans{\prod x{A_{1}}{\app{B_{1}}x}}
\end{array}
\]

\end{itemize}
\end{proof}
\begin{lemma}
\label{lem:inverse-equivalence} $ $
\begin{enumerate}
\item If $M\equiv_{\beta R}N$ then $\btrans M\equiv_{\beta}\btrans N$
\item If $A\equiv_{\beta R}B$ then $\bttrans A\equiv_{\beta}\bttrans B$
\end{enumerate}
\end{lemma}
\begin{proof}
Follows from Lemma \ref{lem:inverse-reduction}.
\end{proof}
Because the forward translation of contexts does not introduce any
type variable, we define the following restriction on contexts.
\begin{definition}[Object context]
\label{def:abstraction-context}  We say that $\Gamma$ is an \emph{object
context} if $\typing[\lpimod S]{\Gamma}A{\Type}$ for all $x:A\in\Gamma$.
If $\Gamma=(x_{1}:A_{1},\ldots,x_{n}:A_{n})$ is an object context,
we define $\bttrans{\Gamma}$ as $(x_{1}:\bttrans{A_{1}},\ldots,x_{n}:\bttrans{A_{n}})$.\end{definition}
\begin{lemma}
\label{lem:inverse-completion} For any $\lpimmod S$ object context
$\Gamma$ and terms $M,A$:
\begin{enumerate}
\item If $\wellformed[\lpimod S]{\Gamma}$ then $\wellformed[\pts{\completion S}]{\bttrans{\Gamma}}$.
\item If $\typing[\lpimod S]{\Gamma}MA:\Type$ then $\typing[\pts{\completion S}]{\bttrans{\Gamma}}{\btrans M}{\bttrans A}$.
\item If $\typing[\lpimod S]{\Gamma}A{\Type}$ then $\typing[\pts{\completion S}]{\bttrans{\Gamma}}{\bttrans A}s$
for some sort $s\in\completion{\set S}$.
\end{enumerate}
\end{lemma}
\begin{proof}
By induction on the derivation. The details of the proof can be found
in the Appendix.
\end{proof}

\subsection{\label{sub:completion-conservativity}Reduction to $\protect\pts S$}

In order to show that $\pts{\completion S}$ is a conservative extension
of $\pts S$, we prove that $\beta$-reduction at the level of $\tau$
terminates. A straightforward proof by induction would fail because
contracting a $\tau$-level $\beta$-redex can create other such redexes.
To solve this, we adapt Tait's \emph{reducibility method} \cite{tait_intensional_1967}.
The idea is to strengthen the induction hypothesis of the proof by
defining a predicate by induction on the type of the term.

\medskip{}

\begin{definition}
\label{def:reducibility}The predicate $\reducible{\Gamma}MA$ is
defined as $\wellformed[\pts S]{\Gamma}$ and $\typing[\pts{\completion S}]{\Gamma}{M:A}s$
for some sort $s$ and:
\begin{itemize}
\item if $\ensuremath{s\not=\tau}$ or $A=s'$ for some $s'\in\set S$ then
$\reducible{\Gamma}MA$ iff $M\reduces_{\beta}^{*}M'$ and $A\reduces_{\beta}^{*}A'$
for some $M',A'$ such that $\typing[\pts S]{\Gamma}{M'}{A'}$,
\item if $s=\tau$ and $A=\prod xBC$ for some $B,C$ then $\reducible{\Gamma}MA$
iff for all $N$ such that $\reducible{\Gamma}NB$, $\reducible{\Gamma}{\app MN}{\subst NxC}$.
\end{itemize}
\end{definition}
Note that recursive definition covers all cases thanks to Theorem
\ref{thm:pts-top-sort-types}. To show that it is well-founded, we
define the following measure of $A$.
\begin{definition}
\label{def:height}If $\wellformed[\pts S]{\Gamma}$ and $\typing[\pts{\completion S}]{\Gamma}As$
then $\set H_{\tau}(A)$ is defined as:
\[
\begin{array}{rclc}
\set H_{\tau}(A) & = & 0 & \mbox{if \ensuremath{s\not=\tau}}\\
\set H_{\tau}(s') & = & 0 & \mbox{if \ensuremath{s=\tau}}\\
\set H_{\tau}(\prod xBC) & = & 1+max(\set H_{\tau}(B)+\set H_{\tau}(C)) & \mbox{if \ensuremath{s=\tau}}
\end{array}
\]
\end{definition}
\begin{lemma}
If $\typing[\pts{\completion S}]{\Gamma,x:B}C{\tau}$ and $\typing[\pts{\completion S}]{\Gamma}NB$
then $\set H_{\tau}(\subst NxC)=\set H_{\tau}(C)$.\end{lemma}
\begin{proof}
By induction on $C$.\end{proof}
\begin{corollary}
Definition \ref{def:reducibility} is well-founded.\end{corollary}
\begin{proof}
The measure $\set H_{\tau}(A)$ strictly decreases in the definition.
\end{proof}
The predicate we defined is compatible with $\beta$-equivalence.
\begin{lemma}
\label{lem:reducibility-term-equiv} If $\reducible{\Gamma}MA$ and
$\typing[\pts{\completion S}]{\Gamma}{M'}A$ and $M\equiv_{\beta}M'$
then $\reducible{\Gamma}{M'}A$.\end{lemma}
\begin{proof}
By induction on the height of $A$.
\begin{itemize}
\item If $\ensuremath{s\not=\tau}$ or $A=s'$ for some $s'\in\set S$ then
$M\reduces_{\beta}^{*}M''$ and $A\reduces_{\beta}^{*}A'$ for some
$M'',A'$ such that $\typing[\pts S]{\Gamma}{M''}{A'}$. By confluence
and subject reduction, $M'\reduces_{\beta}^{*}M'''$ such that $\typing[\pts S]{\Gamma}{M'''}{A'}$.
\item If $s=\tau$ and $A=\prod xBC$ for some $B,C$ then for all $N$
such that $\reducible{\Gamma}NB$, $\reducible{\Gamma}{\app MN}{\subst NxC}$.
By induction hypothesis, $\reducible{\Gamma}{\app{M'}N}{\subst NxC}$.
Therefore $\reducible{\Gamma}{M'}{\prod xBC}$.\hfill{}\qedhere
\end{itemize}
\end{proof}
\medskip{}

\begin{lemma}
\label{lem:reducibility-type-equiv} If $\reducible{\Gamma}MA$ and
$\typing[\pts{\completion S}]{\Gamma}{A'}s$ and $A\equiv_{\beta}A'$
then $\reducible{\Gamma}M{A'}$.\end{lemma}
\begin{proof}
By induction on the height of $A$.
\begin{itemize}
\item If $\ensuremath{s\not=\tau}$ or $A=s'$ for some $s'\in\set S$ then
$M\reduces_{\beta}^{*}M'$ and $A\reduces_{\beta}^{*}A''$ for some
$M',A''$ such that $\typing[\pts S]{\Gamma}{M'}{A''}$. By conversion,
$\typing[\pts{\completion S}]{\Gamma}M{A'}$, so by subject reduction
$\typing[\pts{\completion S}]{\Gamma}{M'}{A'}$. By confluence, subject
reduction, and conversion, $A'\reduces_{\beta}^{*}A'''$ such that
$\typing[\pts S]{\Gamma}{M'}{A'''}$.
\item If $s=\tau$ and $A=\prod xBC$ for some $B,C$ then for all $N$
such that $\reducible{\Gamma}NB$, $\reducible{\Gamma}{\app MN}{\subst NxC}$.
By product compatibility, $A'=\prod x{B'}{C'}$ such that $B\equiv_{\beta}B'$
and $C\equiv_{\beta}C'$. By induction hypothesis, $\reducible{\Gamma}{\app MN}{\subst Nx{C'}}$.
Therefore $\reducible{\Gamma}M{\prod x{B'}{C'}}$.\hfill{}\qedhere
\end{itemize}
\end{proof}
\medskip{}
We extend the definition of the inductive predicate to contexts and
substitutions before proving the main general lemma.
\begin{definition}
If $\wellformed[\pts{\completion S}]{\Gamma}$, $\wellformed[\pts S]{\Gamma'}$,
and $\sigma$ is a substitution for the variables of $\Gamma$, then
$\reducible{\Gamma'}{\sigma}{\Gamma}$ when $\reducible{\Gamma'}{\sigma(x)}{\sigma(A)}$
for all $(x:A)\in\Gamma$.\end{definition}
\begin{lemma}
\label{lem:completion-reducible} If $\typing[\pts{\completion S}]{\Gamma}MA:s$
then for any context $\Gamma'$ and substitution $\sigma$ such that
$\wellformed[\pts S]{\Gamma'}$ and $\reducible{\Gamma'}{\sigma}{\Gamma}$,
$\reducible{\Gamma'}{\sigma(M)}{\sigma(A)}$.\end{lemma}
\begin{proof}
By induction on the derivation of $\typing[\pts{\completion S}]{\Gamma}MA$.
The details of the proof can be found in the Appendix.\end{proof}
\begin{corollary}
\label{cor:completion-conservativity} Suppose $\wellformed[\pts S]{\Gamma}$
and either $\typing[\pts S]{\Gamma}As$ or $A=s$ for some $s\in\set S$.
If $\typing[\pts{\completion S}]{\Gamma}MA$ then $M\reduces_{\beta}^{*}M'$
such that $\typing[\pts S]{\Gamma}{M'}A$.\end{corollary}
\begin{proof}
Taking $\sigma$ as the identity substitution, there are terms $M'$
and $A'$ such that $M\reduces_{\beta}^{*}M'$ and $A\reduces_{\beta}^{*}A'$
and $\typing[\pts S]{\Gamma}{M'}{A'}$. If $A=s\in S$ then $A'=s$
and we are done. Otherwise by conversion we get $\typing[\pts S]{\Gamma}{M'}A$.
\end{proof}
We now have all the tools to prove the main theorem.
\begin{theorem}[Conservativity]
\label{thm:conservativity}  For any $\Gamma$-type $A$ of $\pts S$,
if there is a term $N$ such that $\typing[\lpimod S]{\ttrans{\Gamma}{}}N{\ttrans A{\Gamma}}$
then there is a term $M$ such that $\typing[\pts S]{\Gamma}MA$.\end{theorem}
\begin{proof}
By Lemma \ref{lem:kind-redex-elimination}, there is a $\lpimmod S$
term $N^{-}$ such that $N\reduces_{\beta}^{*}N^{-}$. By subject
reduction, $\typing[\lpimod S]{\ttrans{\Gamma}{}}{N^{-}}{\ttrans A{\Gamma}}$.
By Lemmas \ref{lem:inverse-completion} and \ref{lem:inverse-translation},
$\typing[\pts{\completion S}]{\Gamma}{\btrans{N^{-}}}A$. By Corollary
\ref{cor:completion-conservativity}, there is a term $M$ such that
$\btrans{N^{-}}\reduces_{\beta}^{*}M$ and $\typing[\pts S]{\Gamma}MA$.
\end{proof}

\section{\label{sec:conclusion}Conclusion}

We have shown that $\lpimod S$ is conservative even when $\pts S$
is not normalizing. Even though $\lpimod S$ can construct more functions
than $\pts S$, it preserves the semantics of $\pts S$. This effect
is similar to various conservative extensions of pure type systems
such as \emph{pure type systems with definitions} \cite{severi_pure_1994},
\emph{pure type systems without the $\Pi$-condition} \cite{severi_pure_1995},
or \emph{predicative (ML) polymorphism} \cite{roux_structural_2014}.
Inconsistency in pure type systems usually does not come from the
ability to type more functions, but from the possible impredicativity
caused by assigning a sort to the type of these functions. It is clear
that no such effect arises in $\lpimod S$ because there is no constant
$\tprod{s_{1}}{s_{2}}{s_{3}}$ associated to the type of illegal abstractions.

One could ask whether the techniques we used are adequate. While the
construction of $\pts{S^{*}}$ is not absolutely necessary, we feel
that it simplifies the proof and that it helps us better understand
the behavior of $\lpimod S$ by reflecting it back into a pure type
system. The relative normalization steps of Section \ref{sub:completion-conservativity}
correspond to the normalization of a simply typed $\lambda$ calculus.
Therefore, it is not surprising that we had to use Tait's reducibility
method. However, our proof can be simplified in some cases. A PTS
is \emph{complete }when it is a completion of itself. In that case,
the construction of $\completion S$ is unnecessary. The translations
$\btrans M$ and $\bttrans A$ translate directly into $\pts S$,
and Section \ref{sub:completion-conservativity} can be omitted. This
is the case for example for the calculus of constructions with \emph{infinite
type hierarchy} ($\pts{C^{\infty}}$) \cite{severi_pure_1994}, which
is the basis for proof assistants such as Coq and Matita.

The results of this paper can be extended in several directions. They
could be adapted to show the conservativity of other embeddings, such
as that of the \emph{calculus of inductive constructions} (CIC) \cite{boespflug_coqine:_2012}.
They also indirectly imply that $\lpimod S$ is weakly normalizing
when $\pts S$ is weakly normalizing because the image of a $\pts S$
term is normalizing \cite{cousineau_embedding_2007}. The strong normalization
of $\lpimod S$ when $\pts S$ is strongly normalizing is still an
open problem. The Barendregt-Geuvers-Klop conjecture states that any
weakly normalizing PTS is also strongly normalizing \cite{geuvers_logics_1993}.
There is evidence that this conjecture is true \cite{barthe_weak_2001},
in which case we hope that its proof could be adapted to prove the
strong normalization of $\lpimod S$. Weak normalization could also
be used as an intermediary step for constructing models by induction
on types in order to prove strong normalization.

\subsection*{Acknowledgments}

We thank Gilles Dowek and Guillaume Burel for their support and feedback,
as well as Fr\'{e}d\'{e}ric Blanqui, Rapha\"{e}l Cauderlier, and the
various anonymous referees for their comments and suggestions on previous
versions of this paper.

\bibliographystyle{plain}
\bibliography{conservativity-embeddings-long}
\newpage{}

\appendix

\part*{Appendix}

\section*{Proof details}
\begin{lemma*}[\ref{lem:inverse-completion}]
 For any $\lpimmod S$ object context $\Gamma$ and terms $M,A$:
\begin{enumerate}
\item If $\wellformed[\lpimod S]{\Gamma}$ then $\wellformed[\pts{\completion S}]{\bttrans{\Gamma}}$.
\item If $\typing[\lpimod S]{\Gamma}MA:\Type$ then $\typing[\pts{\completion S}]{\bttrans{\Gamma}}{\btrans M}{\bttrans A}$.
\item If $\typing[\lpimod S]{\Gamma}A{\Type}$ then $\typing[\pts{\completion S}]{\bttrans{\Gamma}}{\bttrans A}s$
for some sort $s\in\completion{\set S}$.
\end{enumerate}
\end{lemma*}
\begin{proof}
By induction on the derivation.
\begin{enumerate}
\item There are 2 cases.

\begin{itemize}
\item $\infer[Empty]{ }{
	\wellformed{\cempty}}$\\
Then $\wellformed{\cempty}$ trivially.\\
\medskip{}

\item $\infer[Declaration]{
	\wellformed{\Gamma} \\
	\typing{\Gamma}{A}{\Type} \\
	x \not\in \cconcat{\Sigma}{\Gamma}}{
	\wellformed{\cdecl{\Gamma}{x}{A}}}$\\
Then $x\not\in\bttrans{\Gamma}$. By induction hypothesis, $\wellformed{\bttrans{\Gamma}}$
and $\typing{\bttrans{\Gamma}}{\bttrans A}s$ for some sort $s\in\completion{\set S}$.
Therefore $\wellformed{\cdecl{\bttrans{\Gamma}}x{\bttrans A}}$.
\end{itemize}
\item There are 4 cases.

\begin{itemize}
\item $\infer[Variable]{
	\wellformed{\Gamma} \\
	(x:A) \in \cconcat{\Sigma}{\Gamma}}{
	\typing{\Gamma}{x}{A}}$\\
By induction hypothesis, $\wellformed{\bttrans{\Gamma}}$.

\begin{enumerate}
\item If $x=\tsort{s_{1}}$ then $A=\ttype{s_{2}}$ and $(s_{1}:s_{2})\in\set A$.
Therefore $\typing{\bttrans{\Gamma}}{s_{1}}{s_{2}}$.
\item If $x=\tprod{s_{1}}{s_{2}}{s_{3}}$ then $A=\prod{\alpha}{\ttype{s_{1}}}{\arr{(\arr{\app{\tterm{s_{1}}}{\alpha}}{\ttype{s_{2}}})}{\ttype{s_{3}}}}$
and $(s_{1},s_{2},s_{3})\in\set R$. Therefore $\typing{\cdecl{\cdecl{\bttrans{\Gamma}}{\alpha}{s_{1}}}{\beta}{\arr{\alpha}{s_{2}}}}{\prod x{\alpha}{\app{\beta}x}}{s_{3}}$,
which implies $\typing{\bttrans{\Gamma}}{(\abs{\alpha}{s_{1}}{\abs{\beta}{(\arr{\alpha}{s_{2}})}{\prod x{\alpha}{\app{\beta}x}}})}{\prod{\alpha}{s_{1}}{\arr{(\arr{\alpha}{s_{2}})}{s_{3}}}}$.
\item Otherwise $(x:A)\in\Gamma$, so $(x:\bttrans A)\in\bttrans{\Gamma}$.
By induction hypothesis, $\wellformed{\bttrans{\Gamma}}$. Therefore
$\typing{\bttrans{\Gamma}}x{\bttrans A}$.\\
\medskip{}

\end{enumerate}
\item $\infer[Application]{
	\typing{\Gamma}{M}{\prod{x}{A}{B}} \\
	\typing{\Gamma}{N}{A}}{
	\typing{\Gamma}{\app{M}{N}}{\subst{N}{x}{B}}}$\\
By induction hypothesis, $\typing{\bttrans{\Gamma}}{\btrans M}{\prod x{\bttrans A}{\bttrans B}}$
and $\typing{\bttrans{\Gamma}}{\btrans N}{\bttrans A}$. Therefore
$\typing{\bttrans{\Gamma}}{\app{\btrans M}{\btrans N}}{\subst{\btrans N}x{\bttrans B}}$.
By Lemma \ref{lem:inverse-substitution}, $\typing{\bttrans{\Gamma}}{\app{\btrans M}{\btrans N}}{\bttrans{\subst NxB}}$\\
\medskip{}

\item $\infer[Abstraction]{
	\typing{\Gamma}{\prod{x}{A}{B}}{\Type} \\
	\typing{\cdecl{\Gamma}{x}{A}}{M}{B}}{
	\typing{\Gamma}{\abs{x}{A}{M}}{\prod{x}{A}{B}}}$\\
By induction hypothesis, $\typing{\bttrans{\Gamma}}{\prod x{\bttrans A}{\bttrans B}}s$
and $\typing{\cdecl{\bttrans{\Gamma}}x{\bttrans A}}{\btrans M}{\bttrans B}$
for some sort $s\in\completion{\set S}$. Therefore $\typing{\bttrans{\Gamma}}{(\abs x{\bttrans A}{\btrans M})}{\prod x{\bttrans A}{\bttrans B}}$.\\
\medskip{}

\item $\infer[Conversion]{
	\typing{\Gamma}{M}{A} \\
	\typing{\Gamma}{B}{\Type} \\
	A\equiv_{\beta R}B}{
	\typing{\Gamma}{M}{B}}$\\
By induction hypothesis, $\typing{\bttrans{\Gamma}}{\btrans M}{\bttrans A}$
and $\typing{\bttrans{\Gamma}}{\bttrans B}s$ for some sort $s\in\completion{\set S}$.
By Lemma \ref{lem:inverse-substitution}, $\bttrans A\equiv_{\beta}\bttrans B$.
Therefore $\typing{\bttrans{\Gamma}}{\btrans M}{\bttrans B}$.
\end{itemize}
\item There are 4 cases.

\begin{itemize}
\item $\infer[Variable]{
	\wellformed{\Gamma} \\
	(x:\Type) \in \cconcat{\Sigma}{\Gamma}}{
	\typing{\Gamma}{x}{\Type}}$\\
Since $\Gamma$ is an object context we must have $x\in\Sigma$, so
$x=\ttype{s_{1}}$ for some $s_{1}\in\set S$. By induction hypothesis,
$\wellformed{\bttrans{\Gamma}}$. By definition, there is a sort $s_{2}\in\completion{\set S}$
such that $(s_{1}:s_{2})\in\completion{\set A}$. Therefore $\typing{\bttrans{\Gamma}}{s_{1}}{s_{2}}$.\\
\medskip{}

\item $\infer[Application]{
	\typing{\Gamma}{M}{\prod{x}{A}{B}} \\
	\typing{\Gamma}{N}{A}}{
	\typing{\Gamma}{\app{M}{N}}{\subst{N}{x}{B}}}$\\
Since $\Gamma$ is an object context and $\app MN$ is not a $\beta$-redex,
we must have $M=\tterm{s_{1}}$ and $\prod xA{B=\arr{\ttype{s_{1}}}{\Type}}$
and $N:\ttype{s_{1}}$ for some $s_{1}\in\set S$. By induction hypothesis,
$\typing{\bttrans{\Gamma}}{\btrans N}{s_{1}}$.\\
\medskip{}

\item $\infer[Product]{
	\typing{\Gamma}{A}{\Type} \\
	\typing{\cdecl{\Gamma}{x}{A}}{B}{\Type}}{
	\typing{\Gamma}{\prod{x}{A}{B}}{\Type}}$\\
By induction hypothesis, $\typing{\bttrans{\Gamma}}{\bttrans A}{s_{1}}$
and $\typing{\cdecl{\bttrans{\Gamma}}x{\bttrans A}}{\bttrans B}{s_{2}}$
for some sorts $s_{1},s_{2}\in\completion{\set S}$. By definition,
there is a sort $s_{3}\in\completion{\set S}$ such that $(s_{1},s_{2},s_{3})\in\completion{\set R}$.
Therefore $\typing{\bttrans{\Gamma}}{(\prod x{\bttrans A}{\bttrans B})}{s_{3}}$.\\
\medskip{}

\item $\infer[Conversion]{
	\typing{\Gamma}{A}{B} \\
	\typing{\Gamma}{B}{\Kind} \\
	B\equiv_{\beta R}\Type}{
	\typing{\Gamma}{A}{\Type}}$\\
We must have $B=\Type$. By induction hypothesis, $\typing{\bttrans{\Gamma}}{\bttrans A}s$
for some sort $s\in\completion{\set S}$.
\end{itemize}
\end{enumerate}
\end{proof}
\begin{lemma*}[\ref{lem:completion-reducible}]
 If $\typing[\pts{\completion S}]{\Gamma}MA:s$ then for any context
$\Gamma'$ and substitution $\sigma$ such that $\wellformed[\pts S]{\Gamma'}$
and $\reducible{\Gamma'}{\sigma}{\Gamma}$, $\reducible{\Gamma'}{\sigma(M)}{\sigma(A)}$.\end{lemma*}
\begin{proof}
By induction on the derivation of $\typing[\pts{\completion S}]{\Gamma}MA$.
\begin{itemize}
\item $\infer[Sort]{
	\wellformed{\Gamma} \\
	(s_1:s_2)\in{\completion{\set{A}}}}{
	\typing{\Gamma}{s_1}{s_2}}$\\
Since $s_{2}:s$, we must have $s_{2}\not=\tau$, so $(s_{1}:s_{2})\in\set A$.
Therefore $\typing[\pts S]{\Gamma'}{s_{1}}{s_{2}}$, which implies
$\reducible{\Gamma'}{s_{1}}{s_{2}}$.\medskip{}

\item $\infer[Variable]{
	\wellformed{\Gamma} \\
	(x:A)\in\cconcat{\Sigma}{\Gamma}}{
	\typing{\Gamma}{x}{A}}$\\
Then $\reducible{\Gamma'}{\sigma(M)}{\sigma(A)}$ by definition of
$\reducible{\Gamma'}{\sigma}{\Gamma}$.\medskip{}

\item $\infer[Application]{
	\typing{\Gamma}{M}{\prod{x}{A}{B}} \\
	\typing{\Gamma}{N}{A}}{
	\typing{\Gamma}{\app{M}{N}}{\subst{N}{x}{B}}}$\\
Without loss of generality, $x\not\in\Gamma'$, so $\sigma(\subst NxB)=\subst{\sigma(N)}x{\sigma(B)}$.
By induction hypothesis, $\reducible{\Gamma'}{\sigma(M)}{\prod x{\sigma(A)}{\sigma(B)}}$
and $\reducible{\Gamma'}{\sigma(N)}{\sigma(A)}$.

\begin{enumerate}
\item If $\typing[\pts{\completion S}]{\Gamma}{\prod xAB}{s_{3}}\not=\tau$
then $\typing[\pts{\completion S}]{\Gamma}A{s_{1}}$ and $\typing[\pts{\completion S}]{\cdecl{\Gamma}xA}B{s_{2}}$
for some $s_{1},s_{2}$ such that $(s_{1},s_{2},s_{3})\in\set S$,
which also means that $\typing[\pts{\completion S}]{\Gamma}{\subst NxB}{s_{2}\not=\tau}$.
By induction hypothesis, $\sigma(M)\reduces_{\beta}^{*}M'$, $\sigma(A)\reduces A'$
and $\sigma(B)\reduces B'$ such that$\typing[\pts{\completion S}]{\Gamma'}{M'}{\prod x{A'}{B'}}$
and $\sigma(N)\reduces_{\beta}^{*}N'$, $\sigma(A)\reduces_{\beta}^{*}A''$
such that $\typing[\pts{\completion S}]{\Gamma'}{N'}{A''}$. By confluence
and subject reduction, we can assume $A'=A''$. Therefore $\typing[\pts{\completion S}]{\Gamma'}{\app{M'}{N'}}{\subst{N'}x{B'}}$.
Since $\subst NxB\reduces_{\beta}^{*}\subst{N'}x{B'}$, this implies
$\reducible{\Gamma'}{\app MN}{\subst NxB}$.
\item Otherwise $\typing{\Gamma}{\prod xAB}{\tau}$. By definition, $\reducible{\Gamma'}{\app{\sigma(M)}{\sigma(N)}}{\subst{\sigma(N)}x{\sigma(B)}}$.\\
\medskip{}

\end{enumerate}
\item $\infer[Abstraction]{
	\typing{\cdecl{\Gamma}{x}{A}}{M}{B} \\
	\typing{\Gamma}{\prod{x}{A}{B}}{s}}{
	\typing{\Gamma}{\abs{x}{A}{M}}{\prod{x}{A}{B}}}$\\
Without loss of generality, $x\not\in\Gamma'$.

\begin{enumerate}
\item If $s\not=\tau$ then by induction hypothesis, $\sigma(A)\reduces_{\beta}^{*}A'$
and $\sigma(B)\reduces_{\beta}^{*}B'$ such that $\typing[\pts S]{\Gamma'}{\prod x{A'}{B'}}s$.
By inversion, $\typing[\pts S]{\Gamma'}{A'}{s_{1}}$ for some $s_{1}\not=\tau$
, so $\reducible{\Gamma}A{s_{1}}$, which implies $\reducible{\cdecl{\Gamma'}xA'}{\sigma}{(\cdecl{\Gamma}xA)}$.
By induction hypothesis, $\sigma(M)\reduces_{\beta}^{*}M'$ and $\sigma(B)\reduces_{\beta}^{*}B''$
such that $\typing[\pts S]{\cdecl{\Gamma'}x{A'}}{M'}{B''}$. By confluence
and subject reduction, we can assume $B'=B''$. Therefore $\typing[\pts S]{\Gamma'}{(\abs x{A'}{M'})}{\prod x{A'}{B'}}$,
which implies $\reducible{\Gamma'}{(\abs xAM)}{\prod xAB}$.
\item If $s=\tau$ then for all $N$ such that $\reducible{\Gamma'}N{\sigma(A)}$,
we have $\reducible{\Gamma'}{(\sigma,N/x)}{(\Gamma,x:A)}$. By induction
hypothesis, $\reducible{\Gamma'}{(\sigma,N/x)(M)}{(\sigma,N/x)(B)}$.
Since $x\not\in\Gamma'$, we have $(\sigma,N/x)(M)=\subst Nx{\sigma(M)}$
and $(\sigma,N/x)(B)=\subst Nx{\sigma(B)}$. Therefore $\reducible{\Gamma'}{\subst Nx{\sigma(M)}}{\subst Nx{\sigma(B)}}$.
By Lemma \ref{lem:reducibility-term-equiv}, $\reducible{\Gamma'}{(\app{(\abs x{\sigma(B)}{\sigma(M)})}N)}{\subst Nx{\sigma(B)}}$.
Therefore $\reducible{\Gamma'}{(\abs x{\sigma(B)}{\sigma(M)})}{\prod xAB}$.\\
\medskip{}

\end{enumerate}
\item $\infer[Product]{
	\typing[\pts{S}]{\Gamma}{A}{s_1} \\
	\typing[\pts{S}]{\cdecl{\Gamma}{x}{A}}{B}{s_2} \\
	(s_1,s_2,s_3)\in{\completion{\set{R}}}}{
	\typing[\pts{S}]{\Gamma}{\prod{x}{A}{B}}{s_3}} $\\
Without loss of generality, $x\not\in\Gamma'$. Since $s_{3}:s$,
we must have $s_{3}\not=\tau$, so $(s_{1},s_{2},s_{3})\in\set R$,
which also means $s_{1}\not=\tau$ and $s_{2}\not=\tau$. By induction
hypothesis, $\sigma(A)\reduces_{\beta}^{*}A'$ such that $\typing[\pts S]{\Gamma'}{A'}{s_{1}}$.
This means that $\wellformed[\pts S]{\cdecl{\Gamma'}x{A'}}$ and $\reducible{\cdecl{\Gamma'}x{A'}}{(\sigma,x/x)}{(\cdecl{\Gamma}xA)}$.
By induction hypothesis, $\sigma(B)\reduces_{\beta}^{*}B'$ such that
$\typing[\pts S]{\Gamma'}{B'}{s_{2}}$. Therefore $\typing[\pts S]{\Gamma'}{(\prod x{A'}{B'})}{s_{3}}$,
which implies $\reducible{\Gamma'}{(\prod x{A'}{B'})}{s_{3}}$.\\
\medskip{}

\item $\infer[Conversion]{
	\typing{\Gamma}{M}{A} \\
	\typing{\Gamma}{B}{s} \\
	A\equiv_\beta B}{
	\typing{\Gamma}{M}{B}}$\\
By induction hypothesis, $\reducible{\Gamma'}{\sigma(M)}{\sigma(A)}$.
Since $A\equiv_{\beta}B$, we have $\sigma(A)\equiv_{\beta}\sigma(B)$.
By Lemma \ref{lem:reducibility-type-equiv}, $\reducible{\Gamma'}{\sigma(M)}{\sigma(A)}$.\end{itemize}
\end{proof}

\end{document}